\documentclass[12pt]{article}

\usepackage{amsmath,amssymb,amsfonts}

\def\op#1{\mathop{{\it\fam0} #1}\limits}

\newcommand{\beq}{\begin{equation}}
\newcommand{\eeq}{\end{equation}}
\newcommand{\ben}{\begin{eqnarray}}
\newcommand{\een}{\end{eqnarray}}
\newcommand{\be}{\begin{eqnarray*}}
\newcommand{\ee}{\end{eqnarray*}}

\newcommand{\gH}{{\mathfrak h}}
\newcommand{\al}{\alpha}
\newcommand{\bt}{\beta}

\newcommand{\la}{\lambda}

\newcommand{\m}{\mu}
\newcommand{\g}{\gamma}
\newcommand{\vr}{\varrho}

\newcommand{\cG}{{\mathfrak g}}

\newcommand{\si}{\sigma}
\newcommand{\Si}{\Sigma}
\newcommand{\w}{\wedge}
\newcommand{\wt}{\widetilde}
\newcommand{\ol}{\overline}
\newcommand{\dr}{\partial}
\newcommand{\ar}{\op\longrightarrow}
\newcommand{\ot}{\otimes}
\newcommand{\id}{{\mathrm{Id}\,}}

\newcommand{\cD}{{\mathcal D}}
\newcommand{\cA}{{\mathcal A}}

\newenvironment{eqalph}{\stepcounter{equation}
\setcounter{equationa}{\value{equation}} \setcounter{equation}{0}

\begin{eqnarray}}{\end{eqnarray}\setcounter{equation}{\value{equationa}}}

\newcounter{equationa}
\newcounter{remark}
\newcounter{theorem}
\newcounter{proposition}
\newcounter{lemma}
\newcounter{corollary}

\def\theremark{\arabic{remark}}
\def\thetheorem{\arabic{theorem}}

\newenvironment{proof}{\noindent {\textit{Proof:}}
}{$\Box$\medskip}
\newenvironment{remark}{\refstepcounter{remark} \medskip \noindent{\bf
Remark \theremark.} }{ \medskip }
\newenvironment{theorem}{\refstepcounter{theorem} \medskip \noindent{\sc
Theorem \thetheorem.}\it}{\medskip }
\newenvironment{proposition}{\refstepcounter{theorem} \medskip\noindent{\sc
Proposition \thetheorem.}\it}{\medskip }

\newenvironment{corollary}{\refstepcounter{theorem} \medskip \noindent{\sc
Corollary \thetheorem.} \it}{\medskip }

\begin{document}

\hbox{}

\begin{center}

{\large\bf Classical Higgs fields}

\bigskip

{\bf G. Sardanashvily} \medskip

Moscow State University, Russia

Lepage Research Institute, Czech Republic

\end{center}

\bigskip

\begin{abstract}

We consider classical gauge theory on a principal bundle $P\to X$
in a case of spontaneous symmetry breaking characterized by the
reduction of a structure group $G$ of $P\to X$ to its closed
subgroup $H$. This reduction is ensured by the existence of global
sections of the quotient bundle $P/H \to X$ treated as classical
Higgs fields. Matter fields with an exact symmetry group $H$ in
such gauge theory are considered in the pairs with Higgs fields,
and they are represented by sections of a composite bundle $Y\to
P/H\to X$, where $Y\to P/H$ is a fiber bundle associated to a
principal bundle $P \to P/H$ with a structure group $H$. A key
point is that a composite bundle $Y\to X$ is proved to be
associated to a principal $G$-bundle $P \to X$. Therefore, though
matter fields possess an exact symmetry group $H\subset G$, their
gauge $G$-invariant theory in the presence of Higgs fields can be
developed. Its gauge invariant Lagrangian factorizes through the
vertical covariant differential determined by a connection on a
principal $H$-bundle $P\to P/H$. In a case of the Cartan
decomposition of a Lie algebra of $G$, this connection can be
expressed in terms of a connection on a principal bundle $P \to
X$, i.e., gauge potentials for a group of broken symmetries $G$.
\end{abstract}

\section{Introduction}

In general, classical gauge theory comprises fields of three
types: gauge potentials, matter fields, and classical Higgs
fields. The latter are attributes of gauge theory in a case of
spontaneous symmetry breaking.

Spontaneous symmetry breaking is a quantum phenomenon when
automorphism of a quantum algebra need not preserve its vacuum
state \cite{pref,sard15}. In this case, we have inequivalent
vacuum states of a quantum system which are classical objects. For
instance, spontaneous symmetry breaking in Standard Model of
particle physics is ensured by the existence of a constant vacuum
Higgs field \cite{nov,SM}.

Classical field theory adequately is formulated in terms of
Lagrangian theory on smooth fiber bundles whose sections are
classical fields \cite{sard08,book09}. Correspondingly, classical
gauge theory is classical field theory on principal and associated
bundles \cite{book09,book00}. In  gauge theory on a principal
bundle $P\to X$, spontaneous symmetry breaking pertains to
reducing its structure Lie group $G$ to a closed subgroup $H$ of
exact symmetries \cite{iva,tra,keyl,sard92}. Such a reduction is
possible iff the quotient fiber bundle $P/H\to X$ admits global
sections $h$ (Theorem \ref{redsub}). These sections can be
interpreted as classical Higgs fields
\cite{iva,tra,sard92,sard06a,tmf14}. They parameterize the
principal reduced subbundles $P^h$ (with a structure group $H$) of
a principal bundle $P$. These subbundles are inequivalent (Remark
\ref{tmp10}) and need not be isomorphic (Theorems \ref{iso0} --
\ref{iso1}).

If a structure group $G$ of a principal bundle $P\to X$ is reduced
to a closed subgroup $H$, then in the framework of this gauge
theory, we can consider matter fields with an exact symmetry group
$H$. They are described by sections $s_h$ of fiber bundles $Y^h$
which possess a typical fiber $V$ endowed with the left action of
a group $H$, and which are associated to principal reduced
subbundles $P^h\subset P$. Because the subbundles $P^h$ for
different $h$ fails to be equivalent, such matter fields can enter
the theory only in a pair with a certain Higgs field $h$. A
problem of describing the set of all pairs $(s_h,h)$ of matter and
Higgs fields thus arises. These pairs are represented by sections
of the composite bundle $Y\to P/H\to X$ (\ref{b3225}), where $Y\to
P/H$ is a fiber bundle with a structure group $H$ and a typical
fiber $V$, and this fiber bundle is associated to a principal
$H$-bundle $P\to P/H$ (Section 5). The geometry of such composite
bundles has been studied in \cite{book09,sard06a,tmf14}. A key
observation is that, for any section $h$ of the quotient bundle
$P/H\to X$, the pull-back $h^*Y$ of a fiber bundle $Y\to P/H$ is a
subbundle $Y^h$ of $Y\to X$ which is associated to a principal
reduced subbundle $P^h\subset P$ with a structure group $H$. Its
sections $s_h$ correspond to matter fields in the presence of a
background Higgs field $h$.

Following \cite{tmf14}, we here prove that a composite bundle
$Y\to X$ is a $P$-associated bundle with a structure group $G$
(Theorem \ref{LL1}). This allows describing matter fields with an
exact symmetry group $H$ in terms of gauge theory on a principal
bundle $P$ (Section 6). A key point is that a Lagrangian of these
matter fields factorizes through the vertical covariant
differential $\wt D$ (\ref{7.10}), determined by an $H$-connection
on a fiber bundle $Y\to P/H$. The restriction $A_h$ of this
connection to a subbundle $Y^h\subset Y$ then becomes an
$H$-connection on this subbundle (Proposition \ref{LL4}), and the
restriction $Y^h$ of the vertical covariant differential $\wt D$
does the differential covariant with respect to a connection $A_h$
(Proposition \ref{tmp40}).

A problem however is that a connection on a fiber bundle $Y\to
P/H$ is not a dynamical variable in gauge theory. We therefore
assume that the Lie algebra of a group $G$ admits the Cartan
decomposition (\ref{g13}). In this case, any $G$-connection on a
principal bundle $P\to X$ yields an $H$-connection on any reduced
subbundle $P^h$ (Theorem \ref{redt}) and, therefore, induces a
desired $H$-connection on a fiber bundle $Y\to P/H$ (Theorem
\ref{tmp50}). On the configuration space (\ref{lvv}), this results
in gauge theory of gauge potentials of a group $G$, of matter
fields with an exact symmetry subgroup $H\subset G$, and of
classical Higgs fields.

For instance, this is the case of gravitation theory on an
oriented 4-dimensional manifold $X$. It is formulated as gauge
theory with spontaneous symmetry breaking on the principal bundle
$LX$ of linear frames tangent to $X$ with a structure group
$GL^+(4,\mathbb R)$ reduced to the Lorentz group $SO(1,3)$
\cite{iva,gg,tmf,sard11}. Global sections of the corresponding
quotient bundle $LX/SO(1,3)$ are pseudo-Riemannian metrics on the
manifold $X$, which are identified with gravitational fields in
General Relativity. The underlying physical reason of this
spontaneous symmetry breaking is both the geometric Equivalence
principle and the existence of Dirac spinor fields with the
Lorentz spin group of symmetries.

\section{Gauge theory on principal bundles}

We consider smooth fiber bundles (of class $C^\infty$). Smooth
manifolds throughout are assumed to be separable, locally compact,
countable at infinity, paracompact topological spaces.

As already mentioned, we formulate classical gauge theory on a
principal bundle
\beq
\pi_P:P \to X
\label{51f1}
\eeq
over an $n$-dimensional manifold $X$ with a structure Lie group
$G$ acting on $P$ on the right fiberwise freely and transitively
\cite{book09,book00}. For brevity, we call $P$ the principal
$G$-bundle. Its atlas
\beq
\Psi_P=\{(U_\al, z_\al),\vr_{\al\bt}\} \label{51f2}
\eeq
is defined by a family of local sections $z_\al$ with $G$-valued
transition functions $\vr_{\al\bt}$, such that
$z_\bt(x)=z_\al(x)\rho_{\al\bt}(x)$, $x\in U_\al\cap U_\bt$.

Because $G$ acts on $P$ on the right, one considers the quotient
bundles
\beq
 T_GP=TP/G,\qquad  V_GP=VP/G \label{b1.205}
\eeq
over $X$. A typical fiber of the bundle $V_GP\to X$ is the right
Lie algebra $\cG_r$ of a group $G$ with the basis $\{e_p\}$ on
which $G$ acts by the adjoint representation. Sections of the
fiber bundles $T_GP\to X$ and $V_GP\to X$ (\ref{b1.205}) are
$G$-invariant vector fields and vertical $G$-invariant vector
fields on $P$, respectively.

In a general setting, a connection on a principal bundle $P\to X$
is defined as a section of a jet bundle $J^1P\to P$, where $J^1P$
is the jet manifold of a fiber bundle $P\to X$ \cite{book00,sau}.
Because connections on a principal bundle are assumed to be
equivariant with respect to the structure group action (for
brevity, we call them $G$-connections), they are identified to
global sections of the quotient bundle
\beq
C=J^1P/G\to X, \label{tmp30}
\eeq
coordinated by $(x^\m,a^p_\m)$. This is an affine bundle modelled
over a vector bundle $T^*X\ot_X V_GP$. Because of the canonical
embedding
\be
C\ar_X dx^\m\ot(\dr_\m + a_\m^p e_p)\in T^*X\op\ot_X T_GP,
\ee
$G$-connections on $P$ can also be represented in terms of
$T_GP$-valued forms
\beq
A=dx^\la\ot (\dr_\la + A_\la^p e_p).  \label{1131}
\eeq

Let $V$ be a manifold admitting a left action of a stricture group
$G$ of the principal bundle $P$ (\ref{51f1}). a fiber bundle
associated to $P$ with a typical fiber $V$ is then defined as the
quotient space
\beq
Y=(P\times V)/G, \qquad (p,v)/G=(pg, g^{-1}v)/G, \qquad g\in G.
\label{tmp31}
\eeq
For brevity, we call it a $P$-associated bundle.

Every atlas $\Psi_P$ (\ref{51f2}) of a principal bundle $P$
determines an atlas
\beq
\Psi_Y=\{(U_\al,\psi_\al)\}, \qquad \psi_\al(x): (z_\al(x),v)/G\to
v, \label{tmp32}
\eeq
of the associated bundle $Y$ (\ref{tmp31}), and endows $Y$ with
fiberwise coordinates $(x^\la,y^i)$.

Every $G$-connection $A$ (\ref{1131}) on a principal bundle $P$
yields a connection
\beq
A=dx^\la\ot(\dr_\la +A^p_\la I^i_p\dr_i) \label{tmp33}
\eeq
on an associated bundle $Y$, where $\{I_p\}$ is a representation
of the Lie algebra $\cG_r$ in a typical fiber $V$.

\section{Reduced structures and Higgs fields}

Let $H$, dim$H>0$, be a closed subgroup of a structure group $G$
(see Remark \ref{52e1} below). There is a composite fiber bundle
\beq
P\to P/H\to X, \label{b3223a}
\eeq
where
\beq
 P_\Si=P\ar^{\pi_{P\Si}} P/H \label{b3194}
\eeq
is a principal bundle with a structure group $H$ and
\beq
\Si=P/H\ar^{\pi_{\Si X}} X \label{b3193}
\eeq
is a $P$-associated bundle with a typical fiber $G/H$ on which a
structure group $G$ acts on the left.

\begin{remark}  \label{52e1} A closed subgroup $H$ of a Lie group $G$
is a Lie group. We consider the quotient space $G/H$ of a group
$G$ with respect to the right action of $H$ on $G$. One can show
that
\beq
\pi_{GH}:G\to G/H \label{ggh}
\eeq
is a principal bundle with a structure group $H$ \cite{ste}. In
particular, if $H$ is a maximal compact subgroup of $G$, then the
quotient space $G/H$ is diffeomorphic to $\mathbb R^m$, and the
fiber bundle (\ref{ggh}) is trivial.
\end{remark}

A structure Lie group $G$ of a principal bundle $P$ is said to be
reduced to its closed subgroup $H$ if the following equivalent
conditions are satisfied:

$\bullet$ a principal bundle $P$ admits the atlas (\ref{51f2})
with $H$-valued transition functions $\vr_{\al\bt}$;

$\bullet$ there exists a principal reduced subbundle $P_H$ of a
principal bundle $P$ with a structure group $H$.

Indeed, if $P_H\subset P$ is a reduced subbundle, then its atlas
(\ref{51f2}), generated by local sections $z_\al$ also is an atlas
of a principal bundle $P$ with $H$-valued transition functions.
Conversely, let (\ref{51f2}) be an atlas of a principal bundle $P$
with $H$-valued transition functions $\vr_{\al\bt}$. For any $x\in
U_\al\subset X$, we define a submanifold $z_\al(x)H\subset P_x$.
These submanifolds constitute an $H$-subbundle of $P$ because
$z_\al(x)H=z_\bt(x)H\vr_{\bt\al}(x)$ on intersections $U_\al\cap
U_\bt$.

\begin{remark}
Principal reduced $H$-subbundles of a principal $G$-bundle
sometimes are called the $G$-structures
\cite{sard92,kob72,zul,gor}. In \cite{kob72,gor}, only reduced
structures of a principal bundle $LX$ of linear frames in the
tangent bundle $TX$ of a manifold X were considered, and the
isomorphism class of these structures was confined to holonomy
automorphisms of $LX$, i.e., to functorial lift onto $LX$ of
diffeomorphisms of a base $X$. A notion of the $G$-structure was
extended to an arbitrary fiber bundle in \cite{zul}, where it was
interpreted as the Klein–-Chern geometry. In a case where the Lie
algebra of a group $G$ admits the Cartan decomposition
(\ref{g13}), the $G$-structure is said to be reduced \cite{god},
and it manifests several additional features (Theorem \ref{redt}).
\end{remark}

A key point is the following \cite{kob}.

\begin{theorem}\label{redsub}
There is one-to-one correspondence
\beq
P^h=\pi_{P\Si}^{-1}(h(X)) \label{510f2}
\eeq
between the principal reduced $H$-subbundles $i_h:P^h\to P$ of a
principal bundle $P$ and the global sections $h$ of the quotient
bundle $P/H\to X$ (\ref{b3193}).
\end{theorem}

The relation (\ref{510f2}) implies that a principal reduced
$H$-subbundle $P^h$ is the restriction $h^*P_\Si$ of the principal
$H$-bundle $P_\Si$ (\ref{b3194}) to a submanifold $h(X)\subset
\Si$. At the same time, every atlas $\Psi_h$ of a principal bundle
$P^h$ generated by a family of its local sections also is an atlas
of a principal $G$-bundle $P$ and an atlas of the $P$-associated
bundle $\Si\to X$ (\ref{b3193}) with $H$-valued transition
functions. Relative to an atlas $\Psi_h$ of a fiber bundle $\Si$,
a global section $h$ of this bundle takes its values in the center
of the quotient space $G/H$.

As already mentioned, we treat global sections of the quotient
bundle $P/H\to X$ as classical Higgs fields in classical gauge
theory \cite{book09,sard92,tmf14}.

Reducing a structure group is not always possible. In particular,
in an above mentioned case of gauge gravitation theory, it occurs
on noncompact manifolds $X$ and on compact manifolds with the zero
Euler characteristic. We note the following fact \cite{ste}.

\begin{theorem} \label{tmp1}  A bundle $Y\to X$
whose typical fiber is diffeomorphic to a manifold $\mathbb R^m$
always admits a global section, and every section of it over a
closed submanifold of a base $X$ can be extended  to the global
one.
\end{theorem}

\begin{corollary} \label{510a1}
A structure group $G$ of a principal bundle $P$ is reducible to a
closed subgroup $H$ if the quotient space $G/H$ is diffeomorphic
to an Euclidean space $\mathbb R^m$.
\end{corollary}

In particular, we always have a reduction of a Lie structure group
$G$ to its maximal compact subgroup $H$ (see Remark \ref{52e1}).
This is the case $G=GL(m,\mathbb C)$, $H=U(m)$ and $G=GL(n,\mathbb
R)$, $H=O(n)$, which are essential for applications.

Let us note that different principal $H$-subbundles $P^h$ and
$P^{h'}$ of a principal $G$-bundle $P$ need not be mutually
isomorphic.

\begin{theorem}\label{iso0}  If the quotient space $G/H$ is diffeomorphic to
an Euclidean space $\mathbb R^m$, then all principal reduced
$H$-subbundles of a principal $G$-bundle $P$ are mutually
isomorphic \cite{ste}.
\end{theorem}

\begin{theorem}\label{iso1}  Let a Lie structure group $G$ of a principal bundle $P$ be reduced to its closed
subgroup $H$. We have the following statements:

$\bullet$ Every vertical automorphism $\Phi$ of a principal bundle
$P$ maps its principal reduced $H$-subbundle $P^h$ to an
isomorphic principal reduced $H$-subbundle $P^{h'}=\Phi(P^h)$.

$\bullet$ Conversely, let two reduced subbundles $P^h$ and
$P^{h'}$ of a principal bundle $P\to X$ be mutually isomorphic,
and let $\Phi:P^h\to P^{h'}$ be their isomorphism over $X$. Then
it is extended to an automorphism of a principal bundle $P$.
\end{theorem}

\begin{proof}
Let
\beq
\Psi^h=\{(U_\al,z^h_\al), \vr^h_{\al\bt}\} \label{510ff}
\eeq
be an atlas of a principal reduced subbundle $P^h$, where
$z^h_\al$ are local sections of $P^h\to X$ and $\vr^h_{\al\bt}$
are transition functions. Given a vertical automorphism $\Phi$ of
a fiber bundle $P$, its subbundle $P^{h'}=\Phi(P^h)$ is provided
with an atlas
\beq
\Psi^{h'}=\{(U_\al,z^{h'}_\al), \vr^{h'}_{\al\bt}\},
\label{510ff1}
\eeq
determined by local sections $z^{h'}_\al =\Phi\circ z^h_\al$. Then
one can easily obtain
\beq
\vr^{h'}_{\al\bt}(x) =\vr^h_{\al\bt}(x), \qquad x\in U_\al\cap
U_\bt, \label{510ff2}
\eeq
i.e., transition functions of the atlas (\ref{510ff1}) take their
values in a subgroup $H$. Conversely, every automorphism $(\Phi,
\id X)$ of principal reduced subbundles $P^h$ and $P^{h'}$ of a
principal bundle $P$ determines an $H$-equivariant $G$-valued
function $f$ on $P^h$ by the relation $pf(p)=\Phi(p)$, $p\in P^h$.
Its extension to a $G$-equivariant function on $P$ is defined as
\be
f(pg)=g^{-1}f(p)g, \qquad p\in P^h, \qquad g\in G.
\ee
The relation $\Phi_P(p)=pf(p)$, $p\in P$, then defines a vertical
automorphism $\Phi_P$ of a principal bundle $P$ whose restriction
to $P^h$ coincides with $\Phi$.
\end{proof}

\begin{remark} \label{tmp10}
In Theorem \ref{iso1},  we can regard a principal $G$-bundle $P$
endowed with the atlas $\Psi^h$ (\ref{510ff}) as a
$P^h$-associated fiber bundle with a structure group $H$ acting on
its typical fiber $G$ on the left. Correspondingly, being equipped
with the atlas $\Psi^{h'}$ (\ref{510ff1}), a principal bundle $P$
is a $P^{h'}$-associated $H$-bundle. However, the $H$-bundles
$(P,\Psi^h)$ and $(P,\Psi^{h'})$ are not equivalent because their
atlases $\Psi^h$ and  $\Psi^{h'}$ fail to be equivalent. Indeed,
the union of these atlases is an atlas
\be
\Psi=\{(U_\al,z^h_\al, z^{h'}_\al), \vr^h_{\al\bt},
\vr^{h'}_{\al\bt}, \vr_{\al\al}=f(z_\al)\}
\ee
with transition functions
\be
 \vr_{\al\al}(x)=f(z_\al(x)), \qquad z^{h'}_\al(x)
=z^h_\al(x)\vr_{\al\al}(x)=(\Phi_P\circ z^h_\al)(x),
\ee
between the corresponding charts $(U_\al,z^h_\al)$ and
$(U_\al,z^{h'}_\al)$ of atlases $\Psi^h$ and $\Psi^{h'}$,
respectively. However transition functions $\vr_{\al\al}$ are not
$H$-valued. At the same time, the equality (\ref{510ff2}) implies
that transition functions of both atlases constitute the same
cocycle. Hence, the $H$-bundles $(P,\Psi^h)$ and $(P,\Psi^{h'})$
are associated. Owing to an isomorphism $\Phi:P^h\to P^{h'}$, we
can write
\be
P=(P^h\times G)/H=(P^{h'}\times G)/H,\qquad (p\times
g)/H=(\Phi(p)\times f^{-1}(p)g)/H.
\ee
For any $\rho\in H$, we then obtain
\be
&& (p\rho,g)/H=(\Phi(p)\rho, f^{-1}(p)g)/H=(\Phi(p),\rho
f^{-1}(p)g)/H \\
&& \qquad =(\Phi(p), f^{-1}(p)\rho'g)/H,
\ee
where
\beq
\rho'= f(p)\rho f^{-1}(p). \label{510f23}
\eeq
Hence, we can treat $(P,\Psi^{h'})$ as a $P^h$-associated bundle
with the same typical fiber $G$ as for $(P,\Psi^h)$, but the
action $g\to\rho' g$ (\ref{510f23}) of a structure group $H$ on a
typical fiber $G$ of the bundle $(P,\Psi^{h'})$ is not equivalent
to its action $g\to\rho g$ on a typical fiber $G$ of the bundle
$(P,\Psi^h)$, because they have different orbits in $G$.
\end{remark}

\section{Reduction of connections}

We present compatibility conditions for connections on a principal
bundle with its reduced structures \cite{book00,kob}.

\begin{theorem} \label{mos176}
Since connections on a principal bundle are equivariant, every
$H$-connection $A_h$ on an $H$-subbundle $P^h$ of a principal
$G$-bundle $P$ is extendible to a $G$-connection on $P$.
\end{theorem}

\begin{theorem} \label{mos177}
Conversely, the connection  $A$ (\ref{1131}) on a principal
$G$-bundle $P$ is reducible to an $H$-connection on a principal
reduced $H$-subbundle $P^h$ of $P$ iff the corresponding global
section $h$ of the quotient bundle $P/H\to X$ associated to $P$ is
an integral section of the associated connection $A$ (\ref{tmp33})
on $P/H\to X$, i.e., $D^Ah=0$, where $D^A$ is the covariant
differential determined by this connection $A$.
\end{theorem}

In particular, a connection on $P$ always is reducible to a
connection on $P^h$ under the following condition
\cite{book09,kob}.

\begin{theorem} \label{redt}
Let the Lie algebra  $\cG$ of a Lie group $G$ be a direct sum
\beq
\cG = {\mathfrak h} \oplus {\mathfrak f} \label{g13}
\eeq
of the Lie algebra ${\mathfrak h}$ of a Lie group $H$ and its
complement $\mathfrak f$ such that we have the condition
$[{\mathfrak h},{\mathfrak f}]\subset {\mathfrak f}$. Let $A$ be a
$G$-connection on a principal bundle $P$. We consider a principal
reduced bundle $P^h$ with an atlas $\Psi_h$, which also is an
atlas of a principal bundle $P$. Then the pull-back $\ol
A_h=h^*A_{\mathfrak h}$ on $P^h$ of the ${\mathfrak h}$-valued
constituent $A_{\mathfrak h}$ of the connection form $A$
(\ref{1131}) written with respect to an atlas $\Psi_h$ is an
$H$-connection on $P^h$.
\end{theorem}

In this case, matter fields with an exact symmetry group $H$ can
be written in the presence of gauge fields with a larger group $G$
of spontaneously broken symmetries (Section 6).

In particular, the decomposition (\ref{g13}) occurs if $H$ is the
Cartan subgroup of $G$, and we therefore call this the Cartan
decomposition.

For instance, in gauge gravitation theory on the principal frame
bundle $LX$, the Lie algebras $\cG=gl(4,\mathbb R)$ and
${\mathfrak h}=so(1,3)$ satisfy the condition (\ref{g13}), and for
a given pseudo-Riemannian metric $h$, a general linear connection
can be decomposed into the sum of Christoffel symbols, the
contorsion tensor, and the nonmetricity tensor. Its first two
terms constitute the Lorentzian connection on a reduced
$SO(1.3)$-subbundle $L^hX\subset LX$, which allows describing
Dirac spinor fields in gravitation theory in the presence of a
general linear connection \cite{gg,tmf,sard11}.

We can also generate connections on principal reduced subbundles
in a different way.

Let $P\to X$ be a principal bundle. For a morphism of manifolds
$\phi:X'\to X$,  the pull-back bundle $\phi^*P\to X'$ is a
principal bundle with the same structure group as of $P$. If $A$
is a connection on a principal bundle $P$, then the pull-back
connection $\phi^*A$ on $\phi^*P$ is a connection on it as on a
principal bundle \cite{book00}. We hence obtain a result important
in what follows \cite{book09,book00}.

\begin{theorem} \label{mos178}
We consider the composite bundle (\ref{b3223a}). Let $A_\Si$ be a
connection on the principal $H$-bundle $P\to\Si$ (\ref{b3194}).
Then for any principal reduced $H$-bundle $i_h:P^h\to P$, the
pull-back connection $i_h^*A_\Si$ on $P^h$ is an $H$-connection on
this bundle
\end{theorem}

Let us note that, as already mentioned in Introduction, a
Lagrangian of matter fields factorizes through the vertical
covariant differential determined just by a connection on a fiber
bundle $P\to P/H$ (Section 6).

\section{Associated bundles and matter fields}

By virtue of Theorem \ref{redsub}, there is one-to-one
correspondence between the principal reduced $H$-subbundles $P^h$
of a principal bundle $P$ and the Higgs fields $h$. Given such a
subbundle $P^h$, let us consider an associated bundle
\beq
Y^h=(P^h\times V)/H \label{510f24}
\eeq
with a typical fiber V admitting the left action of an exact
symmetry group H. Its sections $s_h$ describe matter fields in the
presence of a Higgs field $h$ and an $H$-connection $A_h$ on a
principal bundle $P^h$.

Different fiber bundles $Y^h$ and $Y^{h'\neq h}$ (\ref{510f24})
are mutually related as follows. If the principal reduced
$H$-subbundles $P^h$ and $P^{h'}$ of a principal $G$-bundle $P$
are isomorphic by virtue of Theorem \ref{iso1}, then the
$P^h$-associated bundle $Y^h$ (\ref{510f24}) also is associated as
\beq
Y^h=(\Phi(p)\times V)/H \label{510f25}
\eeq
to a subbundle $P^{h'}$. If its typical fiber $V$ admits the
action of a whole group $G$, then the $P^h$-associated  bundle
$Y^h$ (\ref{510f24}) also is $P$-associated,
\be
Y^h=(P^h\times V)/H= (P\times V)/G.
\ee
Such $P$-associated  bundles are isomorphic as $G$-bundles, but
not equivalent as $H$-bundles, because transition functions
between their atlases are not $H$-valued (see Remark \ref{tmp10}).

For example, in gauge gravitation theory on a manifold $X$, the
tangent bundle $T$X treated for a given pseudo-Riemannian metric
$h$ as a $L^hX$-associated bundle is a fibering into copies of the
Minkowski space $M^hX$. However for different pseudo-Riemannian
metrics $h$ and $h'$, the fiber bundles $M^hX$ and $M^{h'}X$ are
not equivalent; in particular, representations of their elements
in terms of $\g$-matrices are not equivalent \cite{sard11,sard98}.

Because different fiber bundles $Y^h$ and $Y^{h'\neq h}$ are not
equivalent and need not be isomorphic, one must consider a
$V$-valued matter field only in pair with a certain Higgs field.
We therefore encounter a problem of description  of a set of all
pairs $(s_h,h)$ of matter and Higgs fields.

To describe matter fields in the presence of different Higgs
fields, we consider the composite bundle (\ref{b3223a}) and a
composite bundle
\beq
Y\ar^{\pi_{Y\Si}} \Si\ar^{\pi_{\Si X}} X, \label{b3225}
\eeq
where $Y\to \Si$ is a $P_\Si$-associated bundle
\beq
Y=(P\times V)/H \label{bnn}
\eeq
with a structure group $H$. For a given global section $h$ of the
fiber bundle $\Si\to X$ (\ref{b3193}) and for the corresponding
principal reduced $H$-bundle $P^h=h^*P$, the fiber bundle
(\ref{510f24}) associated to $P^h$ is the restriction
\beq
Y^h=h^*Y=(h^*P\times V)/H \label{b3226}
\eeq
of a fiber bundle $Y\to\Si$ to $h(X)\subset \Si$.

One can then prove the following statements
\cite{book09,sard06a,tmf14}.

\begin{proposition} \label{tmp11}
Every global section $s_h$ of the fiber bundle $Y^h$ (\ref{b3226})
is a global section of the composite fiber bundle (\ref{b3225})
projected onto a section $h=\pi_{Y\Si}\circ s$ of a fiber bundle
$\Si\to X$. Conversely, any global section $s$ of the composite
fiber bundle  $Y\to X$ (\ref{b3225}), when projected onto a
section $h=\pi_{Y\Si}\circ s$ of a fiber bundle $\Si\to X$, takes
its values in the subbundle $Y^hY$ (\ref{b3226}). Hence, there is
one-to-one correspondence between the sections of the fiber bundle
$Y^h$ (\ref{510f24}) and those of the composite bundle
(\ref{b3225}) that cover h.
\end{proposition}

\begin{proposition} \label{tmp35}
An atlas
\beq
\Psi_{P\Si}=\{(U_{\Si \iota},z_\iota), \vr_{\iota\kappa}\}
\label{aaq0}
\eeq
of a principal $H$-bundle $P\to \Si$ and correspondingly of an
associated bundle $Y\to\Si$ defines an atlas
\beq
\Psi^h=\{(\pi_{P\Si}(U_{\Si \iota}),z_\iota\circ h),
\vr_{\iota\kappa}\circ h\} \label{aaq}
\eeq
of a reduced $H$-subbundle $P^h$ and hence of an associated bundle
$Y^h$, which also is an atlas of a principal bundle $P$ with
$H$-valued transition functions.
\end{proposition}

Given an atlas $\Psi_P$ of a principal bundle $P$, which
determines the atlas of the associated bundle $\Si\to X$
(\ref{b3193}), and an atlas $\Psi_{Y\Si}$ of a fiber bundle $Y\to
\Si$,  we can endow the composite fiber bundle $Y\to X$
(\ref{b3225}) with the corresponding coordinate system
$(x^\la,\si^m,y^i)$, where $(\si^m)$ are fiberwise coordinates on
$\Si\to X$ and $(y^i)$ are those on $Y\to\Si$.

\begin{proposition} \label{LL4}
Let
\beq
A_\Si=dx^\la\ot(\dr_\la + \cA^a_\la e_a) + d\si^m\ot(\dr_m +
\cA^a_m e_a) \label{510f04}
\eeq
be a principal connection on a principal $H$-bundle $P\to \Si$,
and let
\beq
A_{Y\Si}=dx^\la\ot(\dr_\la + \cA^a_\la(x^\m,\si^k) I_a^i\dr_i) +
d\si^m\ot(\dr_m + \cA^a_m (x^\m,\si^k)I^i_a\dr_i) \label{510f05}
\eeq
be an associated connection on $Y\to\Si$, where $\{I_a\}$ is a
representation of the right Lie algebra $\gH_r$ of a group $H$ in
$V$. Then, for any $H$-subbundle $Y^h\to X$ of a composite bundle
$Y\to X$, the pull-back connection
\beq
A_h=h^*A_{Y\Si}= dx^\la\ot[\dr_\la + (\cA^a_m (x^\m,h^k)\dr_\la
h^m +\cA^a_\la(x^\m,h^k)) I_a^i\dr_i], \label{510f06}
\eeq
on $Y^h$ is associated to the pull-back connection $h^*A_\Si$ on a
principal reduced $H$-subbundle $P^h$ in Theorem \ref{mos178}.
\end{proposition}

Every connection $A_\Si$ (\ref{510f04}) on a fiber bundle
$Y\to\Si$ determines the first-order differential operator
\beq
\wt D: J^1Y\to T^*X\op\otimes_Y V_\Si Y, \qquad \wt D=
dx^\la\otimes(y^i_\la- \cA^i_\la -\cA^i_m\si^m_\la)\dr_i,
\label{7.10}
\eeq
acting on a composite bundle $Y\to X$, where $V_\Si Y$ is the
vertical tangent bundle to a fiber bundle $Y\to\Si$. It is called
the vertical covariant differential, and has the following
important property.

\begin{proposition} \label{tmp40}
For any section $h$ of a fiber bundle $\Si\to X$, the restriction
of the vertical differential $\wt D$ (\ref{7.10}) on the fiber
bundle $Y^h$ (\ref{b3226}) coincides with the covariant
differential $D^{A_h}$ on $Y^h$ with respect to the pull-back
connection $A_h$ (\ref{510f06}).
\end{proposition}

We thus find that those are sections of the composite bundle $Y\to
X$ (\ref{b3225}) that describe the pairs $(s_h, h)$ of matter and
Higgs fields in classical gauge theory with spontaneous symmetry
breaking.

The following fact is essential when constructing gauge theory
with spontaneous symmetry breaking \cite{tmf14,higgs13}.

\begin{theorem} \label{LL1}  The composite bundle $Y\to X$ (\ref{b3225})
is a $P$-associated bundle whose structure group is $G$ and whose
typical fiber is an $H$-bundle
\beq
W=(G\times V)/H, \label{wes}
\eeq
associated to a principal $H$-bundle $G\to G/H$ (\ref{ggh}).
\end{theorem}

\begin{proof}
Let us represent a fiber bundle $P\to X$ as a $P$-associated
bundle
\be
P=(P\times G)/G, \qquad (pg',g)=(p,g'g), \qquad p\in P, \qquad
g,g'\in G,
\ee
whose typical fiber is a group space of $G$ on which a group $G$
acts by left multiplications. We can then represent the quotient
(\ref{bnn}) in a form
\be
&& Y=(P\times (G\times V)/H)/G, \\
&& (pg',(g\rho,v))= (pg',(g,\rho v))=(p,g'(g,\rho v))=(p,(g'g,\rho v)).
\ee
Therefore, $Y$ (\ref{bnn}) is a $P$-associated bundle with the
typical fiber $W$ (\ref{wes}) on which the structure group $G$
acts according to the law
\beq
g': (G\times V)/H \to (g'G\times V)/H. \label{iik}
\eeq
This is the so-called induced representation of a group $G$ by its
subgroup $H$ \cite{mack}. Given an atlas $\{(U_a,z_a)\}$ of a
principal $H$-bundle $G\to G/H$, the induced representation
(\ref{iik}) takes a form
\be
&& g': (\si,v)=(z_a(\si),v)/H\to (\si',v')=(g'z_a(\si),v)/H= \\
&& \qquad (z_b(\pi_{GH}(g'z_a(\si)))\rho',v)/H=
(z_b(\pi_{GH}(g'z_a(\si))),\rho'v)/H,\\
&& \rho'=z_b^{-1}(\pi_{GH}(g'z_a(\si)))g'z_a(\si)\in
H, \quad \si\in U_a, \quad  \pi_{GH}(g'z_a(\si))\in U_b.
\ee
For example, if $H$ is the Cartan subgroup of $G$, then the
induced representation (\ref{iik}) is a known nonlinear
realization of the group $G$ \cite{book09,col,jos}.
\end{proof}

\section{Lagrangian of matter fields}

Propositions \ref{LL4} -- \ref{tmp40} and Theorem \ref{LL1} imply
the following peculiarity of formulating Lagrangian gauge theory
with spontaneous symmetry breaking \cite{tmf14,higgs14}.

Let $P\to X$ be a principal bundle whose structure group $G$ is
reduced to a closed subgroup $H$. Let $Y$ be the
$P_\Si$-associated bundle (\ref{bnn}). A total configuration space
of gauge theory of $G$-connections on $P$ in the presence of
matter and Higgs fields is
\beq
J^1C\op\times_X J^1Y, \label{510f00}
\eeq
where $C$ is the quotient bundle (\ref{tmp30}) and $J^1Y$ is the
manifold of jets of a fiber bundle $Y\to X$. A total Lagrangian on
the configuration space (\ref{510f00}) is a sum
\beq
L_\mathrm{tot}=L_A +L_\mathrm{m} + L_\si \label{510f01}
\eeq
of a gauge field Lagrangian $L_A$, a matter field Lagrangian
$L_\mathrm{m}$,  and a Higgs field Lagrangian $L_\si$.

Because we do not specify gauge and Higgs fields and because their
Lagrangians can take rather different forms depending on a model,
for instance, in gauge gravitation theory and in Yang--Mills
theory, we here consider only a matter field Lagrangian
$L_\mathrm{m}$. By virtue of Proposition \ref{tmp40}, it
factorizes as
\beq
L_{\rm m}:J^1C\op\times_XJ^1Y\op\to^{\wt D}T^*X\op\ot_Y V_\Si
Y\to\op\w^nT^*X \label{lvv2}
\eeq
through the vertical differential $\wt D$ (\ref{7.10}). Moreover,
we can demonstrate that such a factorization is a necessary
condition for the gauge invariance of $L_{\rm m}$ under
automorphisms of a principal $G$-bundle $P\to X$ \cite{book09}.

However a problem is that the connection $A_\Si$ (\ref{510f04}) on
a fiber bundle $Y\to P/H$, which determines $\wt D$, is not a
dynamical variable in gauge theory. We therefore assume that the
Lie algebra of a group $G$ admits the Cartan decomposition
(\ref{g13}). In this case, any $G$-connection $A$ on a principal
bundle $P\to X$ determines an $H$-connection $\ol A_h$ on every
reduced subbundle $P^h$ (Theorem \ref{redt}). We can then prove
the following theorem \cite{tmf14,higgs14}.

\begin{theorem} \label{tmp50}
We have the $H$-connection $A_\Si$ (\ref{510f04}) on a fiber
bundle $Y\to P/H$ whose restriction $A_h=h^*A_\Si$ to the
$P^h$-associated bundle $Y^h$ coincides with an $H$-connection
$\ol A_h$ generated on $P^h$  by a connection $A$ on a principal
$G$-bundle $P\to X$.
\end{theorem}

\begin{proof}
Let a principal reduced subbundle $P^h\subset P$, be given, and
let $\ol A_h$ be an $H$-connection on $P^h$ in Theorem \ref{redt}
generated by the $G$-connection $A$ on a principal bundle $P\to
X$. By virtue of Theorem \ref{mos176},  we can extend this
connection to a $G$-connection on $P$ for which $h$ is an integral
section of the associated connection
\be
\ol A_h=dx^\la\ot(\dr_\la +  A_\la^p J^m_p\dr_m)
\ee
on a $P$-associated bundle $\Si\to X$. With respect to the atlas
$\Psi^h$ (\ref{510ff}) of a principal bundle $P$ with $H$-valued
transition functions, a Higgs field $h$ takes values in the center
of a homogeneous space  $G/H$, and a connection $\ol A_h$ is
\beq
\ol A_h=dx^\la\ot(\dr_\la + A_\la^a e_a). \label{gyu}
\eeq
We then obtain
\beq
A=\ol A_h +\Theta= dx^\la\ot(\dr_\la + A_\la^a e_a) +\Theta_\la^b
dx^\la\ot e_b, \label{00a}
\eeq
where $\{e_a\}$ is a basis for the right Lie algebra ${\mathfrak
h}_r$ and $\{e_b\}$ is a basis of its complement ${\mathfrak
f}_r$. The decomposition (\ref{00a}) with respect to an arbitrary
atlas of a principal bundle $P$ takes a form
\be
A=\ol A_h +\Theta, \qquad \Theta=\Theta_\la^p dx^\la\ot e_p,
\ee
and satisfies the relation $\Theta_\la^p J_p^m=D^A_\la h^m$, where
$D^A_\la$ are the covariant derivatives with respect to the
associated connection $A$ on a fiber bundle $\Si\to X$. Let us
consider the covariant differential
\be
D=D^m_\la dx^\la\ot\dr_m=(\si_\la^m- A^p_\la J_p^m)dx^\la\ot\dr_m
\ee
with respect to the associated connection $A$ on $\Si\to X$. We
can represent this differential as a $V\Si$-valued form on the jet
manifold $J^1\Si$ of a fiber bundle $\Si\to X$. Because the
decomposition (\ref{00a}) holds for any section $h$ of a fiber
bundle $\Si\to X$, we obtain a $V_GP$-valued form
$\Theta=\Theta_\la^p dx^\la\ot e_p$ on $J^1\Si$, which satisfies
the equation
\beq
\Theta_\la^p J_p^m=D_\la^m. \label{nmb}
\eeq
As a result, we obtain a $V_GP$-valued form
\be
A_H=dx^\la\ot(\dr_\la + (A_\la^p -\Theta_\la^p)e_p)
\ee
on $J^1\Si$ whose restriction to every submanifold $J^1h(X)\subset
J^1\Si$ is the connection $\ol A_h$ (\ref{gyu}) written with
respect to the atlas $\Psi^h$ (\ref{aaq}). Because the
decomposition (\ref{00a}) holds, the equation (\ref{nmb}) admits a
solution for any $G$-connection $A$. We therefore have the
$V_GP$-valued form
\beq
A_H=dx^\la\ot(\dr_\la + (a^p_\la- \Theta_\la^p)e_p) \label{ljl}
\eeq
on the product $J^1\Si\times_X J^1C$ such that for any connection
$A$ and for any Higgs field $h$, the restriction of $A_H$
(\ref{ljl}) to
\be
J^1h(X)\times A(X)\subset J^1\Si\op\times_X J^1C
\ee
is the  connection $\ol A_h$ (\ref{gyu}), written with respect to
the atlas $\Psi^h$ (\ref{aaq}). Now let $A_\Si$ (\ref{510f04}) be
a connection on a principal $H$-bundle $P\to \Si$. This connection
determines a $V_\Si Y$-valued form
\beq
\wt \cD=dx^\la\ot(y^i_\la - (\cA^a_m\si_\la^m
+\cA^a_\la)I_a^i)\dr_i \label{510fh}
\eeq
(the covariant differential (\ref{7.10})) on the configuration
space (\ref{510f00}). We now assume that, for a given connection
$A$ on a principal $G$-bundle $P\to X$, the pull-back connection
$A_h=h^*A_{Y\Si}$ (\ref{510f06}) on $Y^h$ coincides with $\ol A_h$
(\ref{gyu}) for any $h\in\Si(X)$. By virtue of Proposition
\ref{tmp40}, we can then define components of the form
(\ref{510fh}) as follows. For a given point
\beq
(x^\la, a^r_\m,a^r_{\la\m}, \si^m,\si^m_\la,y^i,y^i_\la) \in
J^1C\op\times_X J^1Y, \label{lvv}
\eeq
let $h$ be a section of a fiber bundle $\Si\to X$ whose jet
$j^1_xh$ in $x\in X$ is $(\si^m,\si^m_\la)$, i.e.,
\be
h^m(x)=\si^m, \qquad \dr_\la h^m(x)=\si^m_\la.
\ee
Let the connection bundle $C$ and the Lie algebra fiber bundle
$V_GP$ be endowed with atlases associated to the atlas $\Psi^h$
(\ref{aaq}). We can then write
\beq
A_h=\ol A_h, \qquad \cA^a_m\si_\la^m +\cA^a_\la = a^a_\la-
\Theta_\la^a. \label{lvv1}
\eeq
These equations for the functions $\cA^a_m$ and $\cA^a_\la$ at the
point (\ref{lvv}) have a solution because $\Theta_\la^a$ are
affine functions in the jet coordinates $\si_\la^m$.
\end{proof}

Having the solution of the equation (\ref{lvv1}), we substitute it
in the covariant differential $\wt D$ (\ref{510fh}) requiring that
the matter field Lagrangian factorizes in the form (\ref{lvv2})
through the form $\wt D$ (\ref{510fh}), called the universal
covariant differential determined by a $G$-connection $A$ on a
principal bundle $P$. As a result, we obtain gauge theory of gauge
potentials of a group $G$, of matter fields with a symmetry
subgroup $H\subset G$ and of classical Higgs fields on the
configuration space (\ref{lvv}).

As mentioned above, an example of a classical Higgs field is a
metric gravitational field in gauge gravitation theory on natural
fiber bundles with the spontaneous symmetry breaking due to the
existence of Dirac spinor fields with the Lorentz spin group of
symmetries or by the geometric equivalence principle
\cite{iva,gg,tmf,sard11}. Describing spinor fields in terms of the
composite bundle (\ref{b3225}), we obtain their Lagrangian
(\ref{lvv2}) in the presence of a general linear connection; this
Lagrangian is invariant under general covariance transformations
\cite{book09,sard11,higgs14a}.

In a more general form, classical Higgs fields also were
considered in theory of spinor fields on the so-called
gauge-natural fiber bundles \cite{pal}.

\end{document}